\documentclass[letterpaper]{article}
\usepackage[]{aaai23}

\usepackage{xcolor}
\usepackage{times}  % DO NOT CHANGE THIS
\usepackage{helvet}  % DO NOT CHANGE THIS
\usepackage{courier}  % DO NOT CHANGE THIS
\usepackage[hyphens]{url}  % DO NOT CHANGE THIS
\usepackage{graphicx} % DO NOT CHANGE THIS
\usepackage{natbib}  % DO NOT CHANGE THIS AND DO NOT ADD ANY OPTIONS TO IT
\usepackage{caption} % DO NOT CHANGE THIS AND DO NOT ADD ANY OPTIONS TO IT
\DeclareCaptionStyle{ruled}{labelfont=normalfont,labelsep=colon,strut=off} % DO NOT CHANGE THIS
\nocopyright
\usepackage{verbatim}
\usepackage{todonotes}

\usepackage{xspace}
\usepackage{amsmath, amsthm, amssymb, thmtools}
\usepackage{algorithm}
\usepackage[noend]{algorithmic}
\renewcommand{\algorithmiccomment}[1]{\bgroup\hfill//~#1\egroup}
\newtheorem{theorem}{Theorem}

\newtheorem{lemma}{Lemma}

\newtheorem{definition}{Definition}
\newtheorem{example}{Example}
\newtheorem{problem}{Problem}
\newtheorem{remark}{Remark}

\newcommand{\ucs}{$\mathsf{UCS}$\xspace}
\newcommand{\eiucs}{$\mathsf{EI}$-$\mathsf{UCS}$\xspace}

\newcommand{\get}{$\mathsf{Get}$-$\theta$\xspace}
\newcommand{\increase}{$\mathsf{Increase}$\xspace}

\newcommand{\beauty}{$\mathsf{BEAUTY}$\xspace}
\newcommand{\beautyps}{$\mathsf{BEAUTY}$-$\mathsf{PS}$\xspace}
\newcommand{\abeauty}{$\mathsf{A}$-$\mathsf{BEAUTY}$\xspace}

\newcommand{\beast}{$\mathsf{BEAST}$\xspace}
\newcommand{\beastps}{$\mathsf{BEAST}$-$\mathsf{PS}$\xspace}
\newcommand{\abeast}{$\mathsf{A}$-$\mathsf{BEAST}$\xspace}

\newcommand{\abnb}{$\mathsf{A}$-$\mathsf{BEAUTY}$\&$\mathsf{BEAST}$\xspace}

\title{A Generalization of the Shortest Path Problem \\ to Graphs with Multiple Edge-Cost Estimates}

\author{	
	Eyal Weiss\textsuperscript{\rm 1},
	Ariel Felner\textsuperscript{\rm 2}, 
	Gal A. Kaminka\textsuperscript{\rm 1}}

\affiliations{
	\textsuperscript{\rm 1} Bar-Ilan University, Israel\\
	\textsuperscript{\rm 2} Ben-Gurion University, Israel\\
	eyal.weiss@biu.ac.il,
	felner@bgu.ac.il,
	galk@cs.biu.ac.il
}

\begin{document}

\maketitle

\begin{abstract}
The shortest path problem in graphs is a cornerstone of AI theory and applications. Existing algorithms generally ignore edge weight computation time.
We present a generalized framework for weighted directed graphs, where edge weight can be computed (estimated) multiple times, at increasing accuracy and run-time expense. 
%
% Option 1:
% This raises a generalized shortest path problem 
%
% Option 2:
This raises several generalized variants of the shortest path problem. We introduce
the problem 
of finding a path with the tightest lower-bound on the optimal cost. 
% that optimizes different aspects of path cost and its uncertainty. 
We then present two complete algorithms for the
generalized problem, and empirically demonstrate their efficacy.
\end{abstract}
\section{Introduction}\label{sec:intro}

The canonical problem of finding the shortest path in a directed, weighted graph is fundamental to artificial intelligence and its applications.
The \emph{cost} of a path in a weighted graph, is the sum of the weights of its edges. 
Informed and uninformed search algorithms for finding  \emph{shortest} (minimal-cost) paths are heavily used in planning, scheduling, 
machine learning, constrained % TODO Journal: satisfaction and 
optimization, and more.  
 
A common assumption made by existing search algorithms is that the edge weights are determined in negligible (or very small constant) time. % Journal: can clarify that we mean also "with small constant)"
However, recent advances challenge this assumption.
This occurs when weights are determined by queries to remote sources, %TODO: specific references for searching in remote sources would be good
or when the graph is massive, and is stored in external memory (e.g., disk). In such cases, additional data-structures and algorithmic modifications are needed to optimize the order in which edges are visited, i.e., optimizing access to external memory~\cite{vitter01,hutchinson03,jabbar08,korf2008linear,korf08tbbfs,korf16,sturtevant2016external}.
Similarly, when edge weights are computed dynamically using learned models, or external procedures, it is beneficial to delay weight evaluation until necessary~\cite{dellin2016unifying,narayanan2017heuristic,mandalika2018lazy,mandalika2019generalized}.

A concrete example serves to illustrate the setting we address. Consider searching for the fastest route between two cities, where edges represent roads, and edge weights represent current travel times, which are queried from an online source (e.g., Google maps).
Even a few milliseconds for each query makes the weight evaluation a significant component in the
search run-time. Travel times can be estimated even more accurately with information from additional sources (e.g., weather conditions, road curvature and elevation), but their use may significantly increase edge weight computation time.

% TODO: 
% I wanted to add "It is possible to locally cache query results, and modify the algorithm to make better use of such a cache, but this only alleviates the problem."  But this would have made the argument worse, as our proposal does not make use of such caches, and if anything, increases the number of queries!  Our approach is really just about edge computation time, not query time.

%The prevalent approach managing weight computation is to delay edge , e.g., by visiting an edge without computing its weight until later, or delaying the edge visit as much as possible~\cite{dellin2016unifying,narayanan2017heuristic,mandalika2018lazy,mandalika2019generalized}. 

We present a novel approach to handling expensive weight computation by allowing the search algorithms 
to incrementally use multiple \textit{weight estimators}, that compute the edge weight with increasing accuracy, but also at increasing computation time. Specifically, we replace edge weights with an ordered set of estimators, each providing a lower and upper bound on the true weight.  A search algorithm may quickly compute loose bounds on the edge weight, and invest more computation on a tighter estimator later in the process.  In the example above, a local database can be queried
quickly to get rough bounds on the travel times (based on the fixed distance and speed limits). Incrementally, online queries and computations can be used as needed
to get more accurate edge weight estimations, at increasing computational expense.
This approach follows the recently suggested concept of dynamic estimation during planning~\cite{weissgeneralization,weiss2022position} and its first implementation in AI planning~\cite{weiss2023planning}.

Having  multiple weight estimators for edges is a proper generalization of standard edge weights, and raises several shortest path problem variants. 
The classic singular edge weight is a special case, of an estimator whose lower- and upper- bounds are equal. However, since the true weight may not be known (even applying the most expensive estimator), other variants of the shortest path problems involve finding paths that have the best bounds on the optimal cost.
%address finding paths whose bounds on the shortest-path cost are optimal in some aspect. 

\begin{comment}
%The approach complements existing approaches to delaying edge weight computation, or accessing edges in external memory.
%
The use of weight estimators raises at least three separate shortest path problem variants: Finding the shortest path lower bound (the \textit{best-case shortest path}), finding the shortest path upper bound (the \textit{worst-case shortest path}), and finding the shortest path tightest bound (the \textit{worst-case guaranteed shortest path}). % (the worst-case guaranteed shortest path). 
\end{comment}

%In particular, we wish to determine whether the cost of a given path between two vertices is optimal, or within some suboptimality bound. 
%We show this requires solving
In this paper, we introduce the \textit{shortest path tightest lower-bound} (SLB) problem, which is to find a path with the tightest lower bound on the optimal cost. SLB is an important shortest-path problem variant in graphs
with multiply-estimated edge weights, since its solution provides a lower bound for the true cost of \emph{any} solution, even when costs are unknown, and furthermore it is key to determining optimal (or bounded-suboptimal) paths in such graphs, as we will discuss. 

We present \beauty, an uninformed search algorithm based on \emph{uniform-cost search} (UCS, a variant of \emph{Dijksra's} algorithm)~\cite{felner2011position}. We then use it to construct an anytime algorithm (\abeauty) 
that provides additional flexibility for trade-offs between time spent on search and time spent on estimation.
Both algorithms are shown to be correct and complete.
%which is guaranteed to solve SLB problems.
%The algorithms and theoretical guarantees are discussed in detail. 
Experiments demonstrate
the dramatic computational savings they offer compared to the baseline which uses the most accurate and expensive estimates in all cases.

\begin{comment}
\begin{itemize}

\item our novel framework \& motivation: agent-dependent costs
%מוטיבציה טבעית כאשר משקולות הגרף והמשערכים הם תלויי סוכן (למשל תלויות בפרמטרים פיזיים /  מאפיינים שלו) , ואז התרחיש של חיפוש goal directed עם משערכים דינמיים הגיוני מאוד. זה טבעי למשל עבור רובוטים שצריכים לבצע משימות וכאשר רוצים למזער זמן ביצוע. נשים לב שגישה רווחת היום היא מזעור בעקיפין באמצעות מאפיינים סטטיים של הבעיה, כמו למשל מרחקים. זה מוביל לכך שמנסים לפתור בעיה מקורבת (קלה יותר), בעוד שאנו מנסים להגדיר את הבעיה המדויקת, ומכאן ברור שכאשר באמת רוצים למזער נניח זמן ביצוע אז עדיף לעבוד ישירות על המטריקה הזו.
\item complete anytime algorithm for two of the more fundamental problems
\end{itemize}    
\end{comment}

\section{Background and Related Work}\label{sec:related}
To put this research in context, consider the \textit{abstract} components of the search run-time $T$, in a manner inspired by~\cite{mandalika2019generalized}:
\begin{equation}\label{eq:efforts}
	T=T_w + T_v = \tau_{w} \times w + \tau_{v} \times v,
\end{equation}
where $T_w$ is the time spent on edge weight computation and $T_v$ is the time spent on vertex (search) operations (e.g., expansion, queue operations, etc.).
$T_w, T_v$ can be decomposed as follows: $w$ is the number of edge weight computations conducted, and $v$ the number of vertices encountered during the search; $\tau_{w}, \tau_{v}$ are respectively the average \textit{edge weight computation time}, and average \textit{vertex search operations} time (for every vertex considered). %, such as expansion and priority queue operations).  

We use this abstract view to examine different algorithmic approaches in terms of their efforts to reduce $v$ or $w$, sometimes trading an increase in one parameter to reduce another. Standard search algorithms assume $\tau_{w}$ is negligible (or a small constant) and so their effort is mostly on reducing $v$. In contrast, algorithms for finding shortest paths in robot configuration spaces must consider settings where $\tau_{w}$ is \textit{high}, since in these applications, edge existence and cost are determined by expensive computations for validating geometric and kinematic constraints.  Thus these algorithms reduce $w$ by explicitly delaying weight computations~\cite{dellin2016unifying,narayanan2017heuristic,mandalika2018lazy,mandalika2019generalized}, even at the cost of increasing~$v$. 
% journal:  vertices model points in robot configuration spaces, and edges model potential transformations between them.
 Related challenges arise in planning, where action costs can be computed by external (lengthy) procedures~\cite{dornhege2012semantic,gregory2012planning,ijcai2017p600}, %journal: ,allen21
 or when multiple heuristics have different run-times~\cite{karpas2018rational}.

There are also approaches that seek to directly reduce $\tau_{w}$ (rather than $w$). % TODO Journal:, though not explicitly as we do in this paper. 
%$\tau_{w}$ can be high 
When the graph is too large to fit in random-access memory, it is stored externally (i.e., disk). External-memory graph search algorithms optimize the memory access patterns for edges (and vertices), so as to make better use of faster memory (caching)~\cite{vitter01,hutchinson03,jabbar08,korf2008linear,korf08tbbfs,korf16,sturtevant2016external}.  This reduces $\tau_{w}$ by amortizing the computation costs, but still assumes a single weight computation per edge. 
 
The approach we take in this paper is complementary to those above.  We consider a case where the weight of each edge can be estimated multiple times, successively, at increasing expense for greater accuracy.
The component $T_w=\tau_{w}\times w$ is then replaced with
\begin{equation}\label{eq:efforts_refined}
T_w=\tau_{w_1}\times w_1+\tau_{w_2}\times w_2\ldots \tau_{w_k}\times w_k,
\end{equation}
% $T_w=(\tau_{w_1}\times n_1+\tau_{w_2}\times n_2\ldots \tau_{w_k}\times n_k)$,
with $\tau_{w_1}<\ldots<\tau_{w_k}$, and possibly, $w_1+\ldots +w_k > w$. However, the total sum in Eq.~\ref{eq:efforts_refined} may be much smaller than the original $\tau_{w}\times w$. While
naturally, $\tau_{w_k}\leq \tau_{w}$, the search algorithm may produce $w_k\ll w$.
The algorithms presented here make use of this to balance search effort and edge evaluation in a refined manner, and thus to reduce overall run-time $T_w$. 

This re-thinking of edge weights is independent of, and complementary to, other extensions to the definition of weights in graphs. 
% Weighted graphs are used in numerous ways. Over the years, the definition of weights have been extended in multiple ways. 
For example, 
scalar weights can be \emph{random}, drawn from a distribution associated with each edge~\cite{frank1969shortest}. %journal: ,kwon2013robust,shahabi2015robust
Fuzzy weights~\cite{okada1994fuzzy} allow quantification of uncertainty by grouping approximate weight ranges to several representative sets. 
Multi-objective weights~\cite{StewartW91} allow each edge to be associated with a vector of different weights, facilitating optimization of multiple objectives. 
% journal: Dynamic graphs~\cite{roditty2004dynamic} have structure and weights that change with time. 
% journal: There are many more such examples.
All of these extensions ignore the weight \emph{computation time}, in contrast to the work reported here.

\section{Shortest Path with Estimated Weights}\label{sec:prob}
%% graph definitions 
A standard \emph{weighted digraph} is a tuple $(V, E, c)$, where $V$ is a set of vertices, $E$ is a set of edges, s.t. $e=(v_i,v_j) \in E$ iff there exists an edge from $v_i$ to $v_j$, and $c:E \to \mathbb{R}^+$ is a cost (weight) function mapping each edge to a non-negative number.
Let $v_i$ and $v_j$ be two vertices in $V$.
A \emph{path} $p=\langle e_1,\dots, e_n \rangle$ from $v_i$ to $v_j$ is a sequence of edges $e_k=(v_{q_k},v_{q_{k+1}})$ s.t. $k\in [1,n]$, $v_i=v_{q_1}$, and $v_j=v_{q_{n+1}}$.
The cost of a path $p$ is then defined to be $c(p):=\sum_{k=1}^{n} c(e_k)$. The \emph{Goal-Directed Single-Source Shortest Path}  ($GDS^3P$) 
problem involves finding a path $\pi$ from a start vertex to a goal vertex, with minimal $c(\pi)$, denoted as $C^*$.

%% cost estimators definitions
We now replace the cost function $c$ by an estimator-generating function $\Theta$, which for every
edge $e$ yields a sequence of estimation procedures, each providing a lower and upper bound on the weight
of the edge (Def.~\ref{def:theta}). The procedures are ordered by \textit{increasing running times} and assumed to \emph{yield increasingly tightening bounds}.
% journal: see remark
\begin{comment}
\begin{remark}\label{rem:info}
	Def.~\ref{def:theta} implies that each estimator can only add information, which is obtained by keeping the tightest bounds  when more than one estimator per edge is utilized.\todo{Remove?}
\end{remark}
\end{comment}

\begin{definition}\label{def:theta}
	A \emph{\bf cost estimators function} for a set of edges $E$, denoted as $\Theta$, maps every edge $e \in E$ to a finite and non-empty sequence of \emph{weight estimation procedures}, 
	\begin{equation}\label{eq:Theta}
	\Theta(e):=(\theta_e^1,\dots,\theta_e^{k(e)}), k(e) \in \mathbb{N},
	\end{equation}
	where {\bf estimator} $\theta_e^i$, if applied, returns lower- and upper- bounds $(l_e^i, u_e^i)$ on $c(e)$, such that $0\leq l_e^i \leq c(e) \leq u_e^i < \infty$). % journal: was \leq\infty 	
	$\Theta(e)$ is ordered by the increasing running time of $\theta_e^i$, and the bounds monotonically tighten, i.e., $[l_e^j, u_e^j] \subseteq  [l_e^i, u_e^i]$ for all $i<j$.
	
% Journal	By definition, $l_e^i\in\mathbb{R}^+\cup\{0\}$, $u_e^i\in\mathbb{R}^+\cup\{\infty\}$.
\end{definition}

%% derived notations
\noindent 
This allows us to define \emph{estimated weighted digraphs}:
\begin{definition}\label{def:est-graph}
	An \emph{{\bf estimated weighted digraph}} is a tuple $G=(V,E,\Theta)$, where $V,E$ are sets	 of vertices and edges, resp., and $\Theta$ is a {\bf cost estimators function} for $E$.
\end{definition} 

A path $p = \langle e_1,...,e_n \rangle$ can now be characterized by the accumulated lower- or upper- bounds on the edges,
resulting from the application of \emph{some} weight estimators (Def.~\ref{def:path-bounds}):

\begin{definition}\label{def:path-bounds}  
Let $\Phi({e})$  be a non-empty subset of estimators from the sequence $\Theta(e)$, for an edge $e$.  
%
%	$\Theta$ is said to be {\bf proper} if the 
%	 following hold for any edge $e \in E$:
%	\begin{enumerate}
	%		\item For any $i$ the estimator $\theta_e^i$ returns at least one finite bound, i.e., $l_e^i \ge 0$, or $u_e^i < \infty$ or both.
	%		\item 
	We denote the \emph{tightest} bounds on $c(e)$, over all estimators in $\Phi(e)$, as $l_{\Phi(e)}$ (maximum lower bound) and $u_{\Phi(e)}$ (minimum upper bound):
	\begin{equation}\label{eq:tight_edge_bounds}
		\begin{split}
			%			&l_e^{\Theta(e)}:=\max \{l_e^1,\dots,l_e^{k(e)}\} \\ % ]]\ge 0,\\  %TODO: This by definition is true
			%			&u_e^{\Theta(e)}:=\min \{u_e^1,\dots,u_e^{k(e)}\} \\ % < \infty. %TODO: This may not be true, but why deal with the edge case?
%			(l_e^{\Phi(e)}, u_e^{\Phi(e)}):=(\max \{l_e^1,\dots,l_e^{k(e)}\}, \min \{u_e^1,\dots,u_e^{k(e)}\})
			&l_{\Phi(e)}:=\max\{l_e^i | \theta_e^i=(l_e^i, u_e^i)\in\Phi(e)\}\\
			&u_{\Phi(e)}:=\min\{u_e^j | \theta_e^j=(l_e^j, u_e^j)\in\Phi(e)\}\\
		\end{split}
	\end{equation}
	%	\end{enumerate}

For a path $p$, let  $\Phi(p) := \bigcup_{e\in p} \Phi({e})$. The \emph{path lower bound} and \emph{path upper bound} of $p$ w.r.t. $\Phi(p)$ follow, respectively, from the tightest edge bounds defined above.
\begin{equation}\label{eq:tighe_path_bounds}
l_{\Phi(p)}:=\sum_{i=1}^{n}l_{\Phi({e_i})},~~~~u_{\Phi(p)}:=\sum_{i=1}^{n}u_{\Phi({e_i})}
\end{equation}
% allowed by $\Phi$, i.e., $\Phi_{e_i}:=\Theta(e_i) \cap \Phi$.
%
%These bounds on the path are only defined for cases with non-empty $\Phi_{e_i}$ for every edge $e_i$ in the path.
%
We denote by $\Phi^*(p)$ the \emph{maximal size} $\Phi(p)$, which includes \emph{\bf all} estimators for edges in $p$. 
\end{definition}

%% problem statement
\noindent
Estimated weighted digraphs and their path bounds generalize the familiar weighted digraphs, which are 
a special case where for every edge $e$, there is a single estimation procedure $\theta_e^1=(c(e),c(e))$ with lower and upper bounds that are equal to the weight $c(e)$. In this special case, a shortest tightly-bounded path $\pi$ in the graph is a solution to a $GDS^3P$ problem.
However, in the general case, multiple estimators exist, and we are not guaranteed that every weight can be estimated precisely, even if all estimators for it are used. Thus, several variants of the shortest path problem exist, which correspond to the tightest bounds for the shortest path.

\begin{comment}
 The tightest bounds for the cost of a path $\pi$ then converge to the standard path cost $c(\pi)$. In this special case, we may then state that $\pi$ is a {\bf $\mathcal{B}$-admissible shortest path} if $c(\pi)$ is bounded by a suboptimality factor $\mathcal{B}$, i.e., 
\begin{equation}\label{eq:suboptimality}
	c(\pi) \leq c^* \times \mathcal{B}
\end{equation}
where $c^*$ is the cost of the shortest path, a solution to a  If $\mathcal{B}=1$, then $\pi$ is a
shortest path.

However, in the general case, the cost $c(\pi)$ of a path $\pi$ may not be known precisely, and thus Inequality~\ref{eq:suboptimality} cannot be shown 
directly. Instead, as $c(\pi) \leq u_{\Phi^*(\pi)}$, we may prove that $\pi$ is $\mathcal{B}$-admissible by showing that $u_{\Phi^*(\pi)} \leq c^* \times \mathcal{B}$.  Still, the optimal cost $c^*$ is also unknown, so we instead compare to $l^*$, the tightest lower bound on the  cost of the shortest 
path (see below). Necessarily, $l^*\leq c^*$, thus showing
\begin{equation}\label{eq:admissibility}
	u_{\Phi^*(\pi)} \leq l^* \times \mathcal{B}
\end{equation}
is sufficient to prove that $\pi$ is $\mathcal{B}$-admissible.

In other words, the key step in identifying $\mathcal{B}$-admissible paths with estimated costs (which, for $\mathcal{B}=1$ are shortest paths)  is finding the tightest lower bound on the cost of the shortest path, $l^*$.  To do this, we
re-define the familiar $GDS^3P$ problem, so that we search for the shortest-path tightest lower bound.
\end{comment}

We focus on the \emph{shortest path tighest lower bound} (SLB) problem (Prob.~\ref{prob:l}). This problem deals with determining
a path that achieves $L^*$---the optimal tightest lower bound on the cost of the shortest path.

\begin{problem}[{\bf SLB}, finding $L^*$]\label{prob:l}
Let $P=(G,v_{s},V_{g})$, where $G$ is an estimated weighted digraph with cost estimators functions $\Theta$, $v_{s} \in V$ is the start (source) vertex and $V_{g} \subset V$ is a set of goal vertices. The \emph{\bf Shortest path tightest Lower Bound} problem (SLB) is to find a path $\pi$ %(called a \emph{\bf solution}) 
from  $v_{s}$ to any goal vertex $v \in V_{g}$, such that $\pi$ has the lowest tightest lower bound of any path from $v_s$ to $v\in V_{g}$, w.r.t. $\Theta$, i.e., $l(\pi)=L^*$ with
\begin{equation}\label{eq:l*}
	L^*:=\min_{\pi'} \{ l_{\Phi^*(\pi')}~|~\pi'~\text{is a path from}~v_s~\text{to}~v\in V_g\}.
\end{equation}
% TODO: we abuse the notation l_\Theta, as l_\Phi was defined for a set \Phi, and \Theta is not a set.

%A {\bf solution} for $P$ is a path $\pi$ from the initial vertex.
%A {\bf shortest path} for $P$ is a solution $\pi$ with minimum cost $c^*$, i.e.,
%\begin{equation}\label{eq:c*}
%c(\pi)=c^*:=\min \{ c(\pi')~|~\pi'~\text{is a solution for}~P \}.
%\end{equation} 
\end{problem}
%TODO: Would be better to say $\pi = \argmin_{\pi'} L(\pi') where L is a path cost function in the tuple (G,v,..., L)
% While the definition of GSM problems explicitly leaves the definition of \emph{shortest} open, it impl

\begin{comment}
%Three definitions of \emph{shortest} lead to separate problems.
\begin{problem}[Best-Case Shortest Path]\label{prob:l}
	Given a GSM problem $P$, find a {\bf shortest path lower bound} solution $\pi$ with tightest 
% Journal:  If no solution exists return $\pi=\emptyset, l^*=\infty$.
\end{problem}
\end{comment}

The use of the $\min$ operator may seem counter-intuitive, as typically the tightest lower bound would be the maximal of all lower bounds.
Indeed, ideally, we should use $l_{\Phi^*(\pi^*)}$, the tightest (maximal) lower bound of the shortest path $\pi^*$. However,
$\pi^*$ is unknown (as the cost function itself is unknown). Thus, instead we have to use $L^*$, the \emph{minimal} tightest lowest bound of \emph{any} path that leads from $v_s$ to a goal vertex. Necessarily, the use of $L^*$ bounds $l_{\Phi^*(\pi^*)}$ from below, % so it is valid for testing $\mathcal{B}$-admissibility, 
and on the other hand it is the best (maximal) lower bound we may use, when the true edge costs are unknown. 
%We remind the reader that in Eq.~\ref{eq:admissibility} we are using $l^*$ to compare against the upper-bound $u_{\Theta}(\pi)$, to check for the $\mathcal{B}$-admissibility of $\pi$. A lower $l^*$ is more effective.

The SLB problem (Problem~\ref{prob:l}) is a generalization of the standard shortest path problem $GDS^3P$ (Thm.~\ref{thm:generality}), and thus its complexity is at least that of $GDS^3P$. 
\begin{theorem}\label{thm:generality}
	Problem~\ref{prob:l} %,~\ref{prob:u} and \ref{prob:B} are 
	generalizes $GDS^3P$ problems.
\end{theorem}

\begin{proof}
	We show that \emph{any} standard $GDS^3P$ problem can be formulated as a special case of SLB. In this special case, each edge has one estimator (namely,~$k(e)=1$ for every~$e$), that returns the exact cost (i.e.,~$l_e^1 = c(e) = u_e^1$), as this implies $L^*=C^*$. % and $\mathcal{B}^*=1$ (in the case of $l^*=u^*$ we set $\mathcal{B}^*=1$ even if $c^*=0$ as there is no uncertainty at all). 
	A solution for an SLB instance as described above has lower bound $L^*$ and will therefore have cost $C^*$, hence by definition it is a shortest path.
\end{proof}

The solution to an SLB problem is important as generally exact edge costs may not be known and then $L^*$ provides a tight lower bound for the cost of \emph{any} solution, thus quantifying cost uncertainty. 
Yet it has additional uses.  For example, $L^*$ is key in determining optimal and bounded-suboptimal shortest path solutions.  
To see this, we recall the definition of admissible solutions: a solution $\pi$ to a $GDS^3P$ problem is said to be a {\bf $\mathcal{B}$-admissible shortest path} if $c(\pi)$ is bounded by a suboptimality factor $\mathcal{B}$, i.e., 
\begin{equation}\label{eq:suboptimality}
c(\pi) \leq C^* \times \mathcal{B}.
\end{equation}
If $\mathcal{B}=1$, then $\pi$ is a shortest path.

In estimated weighted digraphs the cost $c(\pi)$ of a path $\pi$ is not known precisely (in the general case), and thus Inequality~\ref{eq:suboptimality} cannot be shown directly. Instead, as $c(\pi) \leq u_{\Phi^*(\pi)}$ holds, we may prove that $\pi$ is $\mathcal{B}$-admissible by showing that $u_{\Phi^*(\pi)} \leq C^* \times \mathcal{B}$.  Still, the optimal cost $C^*$ is also unknown, so we instead compare to $L^*$ (the SLB solution). Necessarily, $L^*\leq C^*$ holds, thus showing
\begin{equation}\label{eq:admissibility2}
	u_{\Phi^*(\pi)} \leq L^* \times \mathcal{B}
\end{equation}
is sufficient to prove that $\pi$ is $\mathcal{B}$-admissible (see Example~\ref{example:graph}).
%illustrates the components of Inequality~\ref{eq:admissibility}. 
Solving SLB is therefore critical to identifying $\mathcal{B}$-admissible paths with estimated costs (which, for $\mathcal{B}=1$ are shortest paths).  

%The next section presents search algorithms for solving SLB. 

%\input{optimal-b}

\begin{example}\label{example:graph}
	Consider an estimated weighted digraph $G=(V,E,\Theta)$, with $V=\{v_0,v_1,v_2,v_3,v_4\}$, and $E=\{e_{01}, e_{02}, e_{14}, e_{21}, e_{23}, e_{24}\}$ (see Fig.~\ref{plot:example1}).
	Here, $\Theta$ is defined by the following estimators: For edge $e_{01}$: $\theta_{e_{01}}^1=(4,4)$. For edge $e_{02}$,  $\theta_{e_{02}}^1=(2,6)$, and $\theta_{e_{02}}^2=(3,5)$.
	For edge $e_{14}$, $\theta_{e_{14}}^1=(1,10)$, $\theta_{e_{14}}^2=(4,6)$. For edge $e_{21}$, $\theta_{e_{21}}^1=(2,3)$, $\theta_{e_{21}}^2=(3,3)$. For edge $e_{23}$, $\theta_{e_{23}}^1=(5,9)$, $\theta_{e_{23}}^2=(7,8)$. Finally, for edge $e_{24}$, $\theta_{e_{24}}^1=(4,6)$.
	Additionally, the true edge costs have the following values: $c_{01}=4, c_{02}=4, c_{14}=5, c_{21}=3, c_{23}=7$ and $c_{24}=6$.

	\begin{figure}[t]
		\centering
		\includegraphics[width=1\columnwidth]{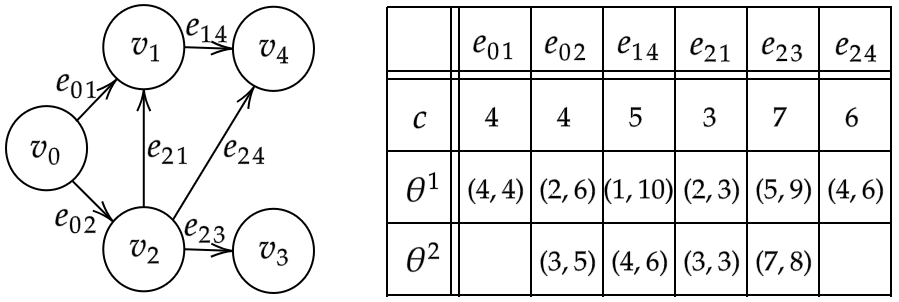}
		\caption{Left: Digraph of Example~\ref{example:graph}. Right: costs and estimates.}
		\label{plot:example1}
	\end{figure}
	
	Given the graph above, we may define the SLB problem $P=(G,v_s,V_g)$ with $v_s=v_0$ and $V_g=\{v_3,v_4\}$, i.e., searching for paths from $v_0$ to either $v_3$, or $v_4$. 
	Then, the unknown optimal cost is $C^*=c(\pi^*)=c_{01}+c_{14}=9$ with $\pi^*=\langle e_{01},e_{14} \rangle$, and the tightest lower bound for $C^*$ is $L^*=l_{\Phi^*(\pi_1)}=l_{02}^2+l_{24}^1=7$ with $\pi_1=\langle e_{02},e_{24} \rangle$ (the SLB solution).
	Then, considering the solution $\pi_2=\langle e_{02},e_{23} \rangle$, one can obtain its tightest upper bound $u_{\Phi^*(\pi_2)}=u_{02}^2+u_{23}^2=13$ and use it to get its admissibility factor $\mathcal{B}(\pi_2)=u_{\Phi^*(\pi_2)}/L^*=13/7$.
\end{example}

%The next section discusses search algorithms for solving SLB. 

\section{Algorithms for SLB}\label{sec:algs}
We present two algorithms for solving the SLB problem. Both  aim at reducing the number of expensive estimators used. The first algorithm, \beauty (Branch\&bound Estimation Applied in \ucs To Yield bottom, Alg.~\ref{alg:beauty}), extends 
\ucs to  dynamically apply cost estimators during a best-first search w.r.t. lower bounds of edge costs.
The second algorithm, \abeauty (Anytime \beauty, Alg.~\ref{alg:abeauty}) uses \beauty in iterations,
such that bounds established in one iteration are used to focus the search in the next, monotonically improving the solution. Both algorithms are proved correct and complete.

\subsection{The First Algorithm: \beauty}\label{subsec:beauty}
Algorithm~\ref{alg:beauty} receives an SLB problem instance and two hyper-parameters $l_{est}, l_{prune}$. For simplicity we will first describe a {\em base case} where $l_{est}, l_{prune}$ are both set to $\infty$, and therefore have no effect and can be ignored. The relevant lines using $l_{est}$ and $l_{prune}$ are colored in blue (Lines 17--19) and should be ignored for now. We will come back to these parameters later.

%\subsubsection{Base Setting}
\paragraph{Base Setting.} \beauty is structurally similar to \ucs. It activates a best-first search process using the standard OPEN and CLOSED lists. Nodes $n$ in OPEN are prioritized by $g_l(n)$ which is, in the base case, always equal to the optimal lower bound to node $n$ along the best known path (similar to using $g(n)$ for ordering nodes in \ucs in regular graphs, which is done according to optimal cost). The {\em best} such node $n$ is chosen for expansion in Line 4, and its successors are added in the loop of Lines 10--23. The main change  of \beauty over \ucs is in the {\em duplicate detection} mechanism performed  when evaluating the cost of a new edge $e$ that connects $n$ to its successor $s$. In \ucs, the exact edge cost $c(e)$ is immediately obtained and used to update the path cost that ends in $s$. In \beauty, we iterate over the different estimators $\theta_e^i$ for edge $e$ (Lines 14--16). In each iteration we set $\tilde{g}_l$ to be the lower bound for the path to $s$ given the current estimator (Line 16). Now such a path can be already pruned earlier if its current lower bound (using the current estimator) will not improve the best known path to $s$ ($g_l(s)$). In that case we will not need to activate the entire set of estimators (in particular, the expensive ones). Thus, if $\tilde{g}_l \ge g_l(s)$, the while statement (Line 14)  ends. Then, ordinary duplicate detection is performed in Lines 19--23. See Example~\ref{example:beauty} for a demonstration of using \beauty in its base setting.

%\subsubsection{Enhanced Setting}
\paragraph{Enhanced Setting.} We now consider the enhanced setting where $l_{est}, l_{prune}$ are set to some constant values (not $\infty$). 
First, $l_{est}$ is used as an upper bound for activating the series of estimators. When a node $n$ has a path lower bound $>l_{est}$ then we no longer activate the series of estimators and only apply the first (cheapest) estimator for edges after $n$ (including the edge to $n$). This is done in Lines 17--18 where we break the loop that further activates estimators on the current edge. 
Second, $l_{prune}$ is used as an upper bound to prune (and not add to OPEN) any node with lower bound $>l_{prune}$ (in a similar manner to {\em bounded cost search}~\cite{SternFBPSG14}). This is done in Line 19.
% in a similar manner to {\em bounded cost search} (BCS)~\cite{SternFBPSG14}. In BCS we are given a constant $C$ and the task is to find a solution with cost $\leq C$. In \beauty, $l_{prune}$ is used as a threshold to prune (and not add to OPEN) any node with lower bound $>l_{prune}$. This is done in Line 19. $l_{prune}$ may be practically used if we already have a path whose best known estimate already equals $l_{prune}$. For example, it is used in the anytime algorithm (\abeauty) discussed below where we have a decreasing series of such upper bounds.
%$l_{est}$ is used as an upper bound for activating the series of estimators. When a node $n$ has a path lower bound $>l_{est}$ then we no longer activate the series of estimators and only apply the first (cheapest) estimator for edges after $n$ (including the edge to $n$). This is done in Lines 17--18 where we break the loop that further activates estimators on the current edge. 
%$l_{est}$ may be practically used if we want to limit the amount of activating costly estimators. It is also used in the anytime algorithm (\abeauty) discussed below where we have a increasing series of such upper bounds.

The purpose of using $L^* \le l_{est}<\infty$ is to avoid applications of redundant (and expensive) estimators. 
Similarly, the purpose of using $L^* \le l_{prune}<\infty$ is to decrease the size of OPEN, which implies less insertion operations and cheaper insert/delete operations.
But since $L^*$ is unknown, setting these hyper-parameters to meaningful values requires prior information.
Practically, such information can be achieved by obtaining a suboptimal solution with $l \ge L^*$, and using it to set $l_{est}=l_{prune}=l$.
This idea is implemented in the anytime algorithm (\abeauty) discussed below.

%\subsubsection{Goal Test and \beautyps}
\paragraph{Goal Test and the Post-Search Procedure \beautyps.} When a solution $\pi$ is found by the $Goal$ function (Line 5), with the path lower bound $l(\pi)$, \beauty % (line~\ref{line:callps}) 
calls \beautyps (post-search procedure, Proc.~\ref{alg:beautyps} below) to iterate over the edges of $\pi$ and tighten the estimations whenever possible, to produce the tightest lower bound $\bar{l}^*$ for $\pi$.
If~$\bar{l}^*=l(\pi)$, namely the path bounds were already tight before \beautyps, then it determines that $\pi$ is optimal and sets $Opt\gets true$. Note that in the base setting when a solution is found it is always already tightly estimated before \beautyps, so no further estimators are applied and $Opt\gets true$.
\beautyps returns $Opt,\underline{l}^*=l(\pi)$ and $\bar{l}^*$, which are then returned by \beauty together with $\pi$ (generated by a path-reconstruction function $trace$).

%Once the goal has been found (Line 5, $Goal(n)$) with its path $\pi$ (calculated with $trace(n)$) then \beautyps is activated to post-process that path. Here we iterate over all edges of $\pi$ and activate the {\em best} available estimator for edges where we did not activate all estimators. This can only occur in the enhanced setting. Such edges are either the last edge (where estimators we pruned to the duplicate detection pruning mechanism of the basic setting) or other edges (due to the pruning of $l_{est}$). 
%\beautyps is a reminiscent of {\em post-smoothing}  or {\em string-pulling} which are done in {\em any-angle pathfinding} and in Robotics~\cite{Theta}. [[AF: the last sentence can be deleted if you so choose.]]

%[[ARIEL EDNED TEXT HERE. Feel free to use/not use any of my text]]

%\subsection{Shortest Path Lower Bound (Problem~\ref{prob:l})}\label{subsec:solving_l}
%\subsection{Beauty}\label{subsec:beauty}
%\paragraph{Algorithm~\ref{alg:beauty}.} \beauty receives an SLB problem instance and the hyper-parameters $l_{est}, l_{prune}$. It works in two stages: First, it utilizes as many estimators as needed whenever it encounters a new edge, in order to determine a best path up to value $l_{est}$. Then, it continues the search with minimal estimations until a solution is found, pruning any path with accumulated lower bound cost greater than $l_{prune}$.  

\begin{algorithm}[t]
	\caption{\beauty}
	\label{alg:beauty}
	\textbf{Input}: Problem~$P=(G,\Theta,v_{s},V_{g})$\\
	\textbf{Parameter}: Thresholds $l_{est}, l_{prune}$\\
	\textbf{Output}: Path $\pi$, $Opt$, bounds $\underline{l}^*, \bar{l}^*$
	\begin{algorithmic}[1] 
		\STATE $g_l(s_0) \gets 0$; OPEN $\gets \emptyset$; CLOSED $\gets \emptyset$
		\STATE Insert $s_0$ into OPEN with $g_l(s_0)$
		\WHILE {OPEN $\neq \emptyset$}
		\STATE $n \gets$ pop node $n$ from OPEN with minimal $g_l(n)$
		\IF {$Goal(n)$}\label{line:goalfound}
		\STATE $l(\pi) \gets g_l(n)$
		\STATE $Opt, \underline{l}^*, \bar{l}^* \gets$ \beautyps\label{line:callps}
		\RETURN {$trace(n), Opt, \underline{l}^*, \bar{l}^*$}
		\ENDIF
		\STATE Insert $n$ into CLOSED 
		\FOR {\textbf{each} successor $s$ of $n$}
		\IF {$s$ not in OPEN $\cup$ CLOSED}
		\STATE $g_l(s) \gets \infty$
		\ENDIF
		\STATE $\tilde{g}_l \gets g_l(n)$
%		\STATE $\theta \gets$ \get$(e=(n,s))$
		\WHILE {$\tilde{g}_l< g_l(s)~\AND$ estimators remain for $e=(n,s)$}
		\STATE $l(e) \gets$ Apply next estimator for $e$
%		\STATE Cache $l$ for $e$
		\STATE $\tilde{g}_l \gets g_l(n) + l(e)$
		{ \color{blue} \IF { $\tilde{g}_l > l_{est}$ }
		\STATE break \COMMENT{exit while loop}
		\ENDIF}
%		\STATE $\theta \gets$ \get$(e)$
		\ENDWHILE
		\IF {$\tilde{g}_l < g_l(s)$ {\color{blue} $\AND~\tilde{g}_l \leq l_{prune}$}}
		\STATE $g_l(s) \gets \tilde{g}_l$
		\IF {$s$ in OPEN}
		\STATE Remove $s$ from OPEN
		\ENDIF
		\STATE Insert $s$ into OPEN with $g_l(s)$
		\ENDIF
		\ENDFOR
		\ENDWHILE
		\RETURN {$\emptyset, false, \infty, \infty$}
	\end{algorithmic}
\end{algorithm}

\floatname{algorithm}{Procedure}
\begin{algorithm}[tbh]
	\caption{\beautyps}
	\label{alg:beautyps}
	\textbf{Input}: \beauty's inputs and variables\\  
	\textbf{Output}: $Opt$, bounds $\underline{l}^*, \bar{l}^*$
	\begin{algorithmic}[1] 
		\STATE $Opt \gets true;~\underline{l}^* \gets l(\pi)$
		\FOR {\textbf{each} edge $e$ in $\pi$}
%		\STATE $\theta \gets$ \get$(e)$
		\IF {estimators remain for $e$}
		\STATE $l \gets$ Apply the best estimator for $e$ 
		\STATE $l(\pi) \gets l(\pi) + l - l(e)$ 
		\STATE $l(e) \gets l$
%		\STATE $\theta \gets$ \get$(e)$
		\ENDIF
		\ENDFOR
		\IF {$l(\pi) > \underline{l}^*$}
		\STATE $Opt \gets false$
		\ENDIF
		\STATE $\bar{l}^* \gets l(\pi)$
		\RETURN {$Opt,~\underline{l}^*,~\bar{l}^*$}
	\end{algorithmic}
\end{algorithm}
\floatname{algorithm}{Algorithm}

%
%Except for the usage of $l_{est}, l_{prune}$ and \beautyps, \beauty is structurally similar to \ucs. The data structures OPEN and CLOSED are priority queues, and $g_l$ is a mapping analogous to $g$ in \ucs. The primary modification is the addition of an estimation loop that takes place in lines 11--21 (including initialization). 

Depending on the hyper-parameters $l_{prune}, l_{est}$, \beauty is \textit{complete} (Lemma~\ref{lemma:conditional_complete_l}), \textit{sound} (Lemma~\ref{lemma:bounds_for_L^*}), and \textit{optimal} (Lemma~\ref{lemma:conditional_optimality_l}).

\begin{lemma}[Conditional Completeness Prob.~\ref{prob:l}]\label{lemma:conditional_complete_l}
	\beauty, provided with $l_{prune} \ge L^*$, is complete.
\end{lemma}

\begin{proof}
\beauty inspects nodes that are removed from OPEN by best-first order w.r.t. lower bound of path cost.
When $l_{prune} = \infty$ is satisfied, no node is pruned, so that every node encountered during the search is inserted into OPEN.
The condition $\tilde{g}_l < g_l(s)$ simply verifies that each node in OPEN points back to the best found path leading to it, but it does not prevent nodes from being inserted.
In this case completeness is assured, as the search is systematic. % journal: (similar to \ucs).

Suppose that a best-first algorithm utilizes all possible estimators per edge it encounters. Then, if a solution exists, a shortest path tightest lower bound $\pi^*$ will necessarily be returned with $L^*$.
Since applying more estimators can only increase (tighten) the lower bound for an edge, it follows that when not all possible estimators per edge are utilized, and a systematic best-first search takes place, then a solution $\pi$ for $P$ ending in a node $n$ will be found, where the key of $n$ in OPEN (the obtained lower bound), immediately before it was removed, must be lower than, or equal to, $L^*$.
This holds regardless of the value of $l_{est}$, that only affects which (and how many) estimators will be applied. Namely, the value of $l_{est}$ may affect which solution $\pi$ is found, but not the fact that such a solution will be found.
Hence, when $l_{prune} \ge L^*$ is satisfied, a solution $\pi$ is necessarily found.
\end{proof}

%\begin{remark}
%	In fact, Lemma~\ref{lemma:conditional_complete_l} can be strengthened to state completeness for $l_{prune}$ with a lower threshold than $l^*$, but for our purposes it is sufficient.
%\end{remark}

\begin{lemma}[Bounds for $L^*$]\label{lemma:bounds_for_L^*}
	\beauty, provided with $l_{prune} \ge L^*$, returns $0 \leq \underline{l}^* \leq L^* \leq \bar{l}^*$, if a solution exists for $P$.
	Furthermore, if~~$l_{est} < L^*$ also holds, then $\underline{l}^* > l_{est}$.
\end{lemma}

\begin{proof}
The proof of Lemma~\ref{lemma:conditional_complete_l} established that when \beauty is called with $l_{prune} \ge L^*$, a solution $\pi$ will be found (when a solution exists), ending in a node $n$, where the key of $n$ in OPEN $g_l(n)$ (the obtained lower bound), immediately before it was removed, satisfies $g_l(n) \leq L^*$. Additionally, $g_l(n) \ge 0$ trivially holds, as each edge lower bound is by definition non-negative.
In line 6 of \beauty $l(\pi) \gets g_l(n)$ is set, then \beautyps is called, which sets $\underline{l}^* \gets l(\pi)$ in Line 1, and then $\underline{l}^*$ is not changed until it is returned.
\beautyps utilizes all unused estimators in the solution $\pi$, by systematically improving estimations for each edge $e$ belonging to $\pi$ using all estimators in $\Theta(e)$. 
Thus the tightest possible lower bound for $\pi$ is obtained and returned as $\bar{l}^*$.
From the optimality of $L^*$ it follows that $\bar{l}^* \ge L^*$.
To sum up, $\underline{l}^*, \bar{l}^*$, that satisfy $0 \leq \underline{l}^* \leq L^* \leq \bar{l}^*$, are returned. 

Let us now consider the case that $l_{est} < L^*$ holds in addition to $l_{prune} \ge L^*$.
Seeking a contradiction, assume that $\underline{l}^* > l_{est}$ is not necessarily satisfied.
This means that for some solution $\pi$, it holds that $\underline{l}^*\leq l_{est}$.
Recall that $\underline{l}^* = g_l(n)$ for the node $n$, which is the last node in the path implied by the solution $\pi$.
Since $l_{est} < L^*$ holds, it must be that each edge in $\pi$ has been estimated using all possible estimators before $n$ is established as a goal node, as for each node $n'$ satisfying the condition $g_l(n') \leq l_{est}$, edges included in the path leading to $n'$ are only denied tight estimation in cases where a better alternative path leading to $n'$ was already found. 
Therefore, the lower bound of $\pi$ cannot be tightened, so $\underline{l}^*=\bar{l}^*$ is satisfied, implying that $\pi$ is optimal with lower bound $L^*$.
But this means that $L^* = \underline{l}^* \leq l_{est} < L^*$. A contradiction.
Hence, $\underline{l}^* > l_{est}$.
%For \beauty called with $l_{prune} \ge l^*, l_{est} < l^*$, let us define $\epsilon:= \min \{\underline{l}^*(\pi) - l_{est}~|~\underline{l}^*(\pi)~\text{is returned for a solution}~\pi\}$.
%Then, $\underline{l}^* \ge l_{est} + \epsilon$ holds.
\end{proof}

\begin{lemma}[Conditional Optimality Prob.~\ref{prob:l}]\label{lemma:conditional_optimality_l}
	\beauty, provided with $l_{prune} \ge L^*$ and $l_{est} \ge L^*$, returns a shortest path tightest lower bound $\pi$ and $\bar{l}^*=L^*$, if a solution exists for $P$.
\end{lemma}

\begin{proof}
Continuing the argument made in the proof of Lemma~\ref{lemma:bounds_for_L^*}, if $l_{prune} \ge L^*$ and $l_{est} \ge L^*$ hold, then the best paths, based on tightest possible estimates, with cumulative lower bounds of up to $l_{est}$ are found, and their terminal nodes are inserted to OPEN.
In particular, the best paths up to $L^*$ (including this value) are found.
From the definition of $L^*$ it follows that there exists a solution $\pi$ with a tight lower bound equal to $L^*$.
Hence, $\pi$, or possibly another solution with the same tight lower bound, is guaranteed to be found when its corresponding goal node is removed from OPEN.
Then, $\bar{l}^*=\underline{l}^*=L^*$ together with $\pi$ are returned.
\end{proof}

\noindent
The implication of Lemmas~\ref{lemma:conditional_complete_l}--\ref{lemma:conditional_optimality_l} is that \emph{SLB problems can be solved optimally using \beauty} by setting $l_{prune}$ and $l_{est}$ to be greater than, or equal to, $L^*$, which can always be achieved by setting them to $\infty$, as Example~\ref{example:beauty} shows.

\begin{example}\label{example:beauty}
	Consider calling \beauty with $l_{est}=l_{prune}=\infty$ (i.e., base setting) on $P$ from Example~\ref{example:graph}. Tracing its run, at the first iteration of the while loop it invokes $\theta_{e_{01}}^1, \theta_{e_{02}}^1$ and $\theta_{e_{02}}^2$ and inserts $v_1, v_2$ to OPEN with keys $4, 3$. At the second iteration $v_2$ is removed from OPEN, $\theta_{e_{21}}^1, \theta_{e_{23}}^1, \theta_{e_{23}}^2, \theta_{e_{24}}^1$ are invoked, and $v_3, v_4$ are inserted to OPEN with keys $10, 7$. At the third iteration $v_1$ is removed from OPEN, $\theta_{e_{14}}^1$ and $\theta_{e_{14}}^2$ are invoked. At the forth iteration $v_4$ is removed from OPEN and \beauty returns $\langle e_{02},e_{24} \rangle, true, 7, 7$.
\end{example}

However, a lower value of $l_{est}$ enables to avoid redundant estimations, where the potential savings grow as $l_{est}$ approaches $L^*$ from above. This motivates the use of \beauty in an iterative framework that gradually increases $l_{est}$ until the optimal solution is found.

\subsection{The Second Algorithm: Anytime \beauty}\label{subsec:any}

The \abeauty algorithm automates the iterative usage of \beauty with increasingly tightened $l_{est}$ and $l_{prune}$ around $L^*$, until the optimal solution is found.
It starts with $l_{est}=0$ and $l_{prune}=\infty$, and each time \beauty terminates it returns $\underline{l}^* > l_{est}$ (Lemma~\ref{lemma:bounds_for_L^*}), which is used as $l_{est}$ in the next call. 
Similarly, the returned $\bar{l}^*$ is a finite value (when a solution exists) that always is greater than, or equal to, $L^*$ (again,  Lemma~\ref{lemma:bounds_for_L^*}). Using the lowest value of $\bar{l}^*$, $l_{prune}$ is monotonically non-increasing.

\begin{algorithm}[t]
	\caption{\abeauty}
	\label{alg:abeauty}
	\textbf{Input}: Problem~$P=(G,\Theta,v_{s},V_{g})$\\
	\textbf{Output}: Path~$\pi$, bound~$\bar{l}^*$
	\begin{algorithmic}[1] 
		\STATE $\underline{l}^* \gets 0$; $\bar{l}^* \gets \infty$; $Opt \gets false$
		\WHILE {not $Opt$}
		\STATE $\pi, Opt, \underline{l}^*, \bar{l} \gets$~\beauty$(P$, $\underline{l}^*, \bar{l}^*)$ 
		\IF {$\pi = \emptyset$}
		\RETURN {$\emptyset, \infty$}
		\ENDIF 
		\IF {$\bar{l} < \bar{l}^*$}
		\STATE $\bar{l}^* \gets\bar{l}$
		\ENDIF
		\STATE Print {$\pi$, $\underline{l}^*$, $\bar{l}^*$}
		\ENDWHILE
		\RETURN {$\pi$, $\bar{l}^*$}
	\end{algorithmic}
\end{algorithm}

The process converges in a finite number of iterations (shown below) and thus assures optimality, while gradually utilizing more estimations, that in turn support better approximations for $L^*$ (which are saved every time an improvement is achieved).  
The estimations are saved between iterations, so that it is not necessary to re-apply estimators.
Technically, this is obtained by defining the next estimator to apply to first look for a saved value and only then turn to unused estimators.
%Overall, this moderates the increase in run-time consumption (due to utilizing more estimations) between subsequent iterations.
Tightened $l_{prune}$ values decrease the size of OPEN, reducing memory consumption and run-time (due to less insertion operations, and cheaper insert/delete operations).

\begin{theorem}[Completeness, Soundness and Optimality Prob.~\ref{prob:l}]\label{thm:complete_opt_l}
	\abeauty is complete. If a solution exists for $P$, then a shortest path tightest lower bound $\pi$ and $L^*$ are returned.
\end{theorem}

\begin{proof}
\abeauty initializes $\underline{l}^*\gets 0$ and $\bar{l}^*\gets\infty$, and then enters a loop that terminates when no solution is found or when the optimal solution is found.
At each iteration of the loop, it calls \beauty with $l_{est}=\underline{l}^*$ and $l_{prune}=\bar{l}^*$.
Due to the initialization, the conditions of Lemmas~\ref{lemma:conditional_complete_l} and~\ref{lemma:bounds_for_L^*} are fulfilled in the first iteration, so that if a solution exists, a solution would be returned by \beauty, with tightened bounds, i.e., $\underline{l}^*>0$ and $L^* \leq \bar{l}^*<\infty$.
In the second iteration (if the optimal solution has yet to be found) the $\underline{l}^*$ and $\bar{l}^*$ found in the first iteration are used again as $l_{est}=\underline{l}^*$ and $l_{prune}=\bar{l}^*$ in the call for \beauty, where again the conditions for both lemmas hold.
Thus $\underline{l}^*$ is guaranteed to monotonically increase with each iteration, and $\bar{l}^*$ can either decrease (but remain at least $L^*$) or stay the same. Hence, the conditions for both lemmas are satisfied in every iteration until termination, i.e., we have established that the conditional completeness of \beauty implies regular completeness for \abeauty, and that $\bar{l}^*$ monotonically non-increases.

To show optimality, we next analyze the increase in $\underline{l}^*$ between subsequent iterations.
Denote $\delta_i:=\underline{l}^*_i-\underline{l}^*_{i-1}$, where $\underline{l}^*_i$ is the value obtained after call $i$ to \beauty.
Note that $\delta_i$ cannot be arbitrarily small values, as they exactly represent the differences between cumulative lower bounds of solutions obtained in subsequent iterations, which are limited to a finite set of values (induced by $\Theta$).
Thus, there exists a constant $\delta_{min}>0$ such $\forall i, \delta_i \ge \delta_{min} $ is satisfied.
Hence, either the optimal solution is found before $\underline{l}^*$ reaches $L^*$, or it is found right after it reaches it (Lemma~\ref{lemma:conditional_optimality_l}), which necessarily occurs after a finite number of iterations.
\end{proof}

\noindent
The proof of Thm.~\ref{thm:complete_opt_l} shows the number of iterations until convergence to optimality is unknown a-priori. Nevertheless, we can set a simple threshold either on the number of iterations or on the convergence implied by $\bar{l}^*/\underline{l}^*$. Once the threshold is crossed, setting both $l_{est}$ and $l_{prune}$ to $\bar{l}^*$ ensures the last iteration. See Example~\ref{example:any_beauty}.

\begin{example}\label{example:any_beauty}
	Consider again the SLB problem $P$ from Example~\ref{example:graph}. 
	When calling \abeauty on $P$, at the first iteration the utilized estimators are $\theta_{e_{01}}^1, \theta_{e_{02}}^1, \theta_{e_{14}}^1, \theta_{e_{14}}^2, \theta_{e_{21}}^1, \theta_{e_{23}}^1$ and $\theta_{e_{24}}^1$, where $\theta_{e_{14}}^2$ is invoked by \beautyps. The algorithm prints $\langle e_{01},e_{14} \rangle, 5, 8$.
	At the second iteration the estimator $\theta_{e_{02}}^2$ is also utilized. The algorithm prints $\langle e_{02},e_{24} \rangle, 7, 7$ and returns $\langle e_{02},e_{24} \rangle, 7$.
\end{example}

\section{Empirical Evaluation}\label{sec:empir}
The theoretical guarantees of \beauty and \abeauty touch on their optimality and completeness, but 
do not provide information as to the run-time savings they offer. We therefore empirically evaluate 
the algorithms in diverse settings, based on AI planning benchmark problems that were modified to have multiple action-cost estimators, so that these induce SLB problems. 

The set of problems was taken from
 %, is part of the main FD repository: https://www.fast-downward.org/}, which is 
a collection of IPC (International Planning Competition) benchmark instances\footnote{See \url{https://github.com/aibasel/downward-benchmarks}.}. 
Starting from the full collection, we first filtered out every domain that didn't offer support for action costs.
Then, for some of the domains we created additional problems by using different configurations of costs.
For all problems and domains, we synthesized three estimators.
Each edge $e$ with cost $c_{old}(e)$ was mapped to a new cost $c_{new}(e)$ that satisfies $c_{new}(e) \ge c_{old}(e) \times f_3$, with $f_3 > f_2 > f_1 \ge 1$, so that $l_e^1:=c_{old} \times f_1, l_e^2:=c_{old} \times f_2, l_e^3:=c_{old} \times f_3$ served as its first, second and third lower bound estimates.
To diversify the estimator sets for different edges, the parameters $f_1, f_2, f_3$ were taken from the sets $f_1 \in \{1, 2, 3\}, f_2 \in \{f_1+1, f_1+2, f_1+3\}, f_3 \in \{f_2+1\}$,
which resulted in nine different configurations.
The choice of configuration was taken according to the result of a simple hash function, that depends on $c_{old}(e)$ and a user-input seed, described as follows: 
\begin{equation}\label{eq:hash}
\text{Hash} = (c_{old}(e) + \text{seed}) \mod 9.
\end{equation}
Then, the configuration was set according to Table~\ref{table:f_config}.
Each problem was run once per seed, where the seeds where taken from the set $[0,8]$, which resulted in 9 instances per problem.
Overall, this resulted in a cumulative set of 914 problem instances, spanning 12 unique domains.
The full list of the domains and problems that were used in the experiments is detailed in~\cite{plandem}.

\begin{table}[t]
	\centering
	\begin{tabular}{r r r r r r r r r r r}
		Hash & 1 & 2 & 3 & 4 & 5 & 6 & 7 & 8 & 9 &  \\ \hline
		$f_1$ & 1 & 2 & 3 & 1 & 2 & 3 & 1 & 2 & 3 &  \\
		$f_2$ & 2 & 3 & 4 & 3 & 4 & 5 & 4 & 5 & 6 &  \\
		$f_3$ & 3 & 4 & 5 & 4 & 5 & 6 & 5 & 6 & 7 &
	\end{tabular}
	\caption{The configuration of $f_1, f_2, f_3$ in Rows 2--4 according to the hash values displayed in Row 1.}
	\label{table:f_config}
\end{table}
 
We note that the configurations depicted in Table~\ref{table:f_config} that are chosen according to the hash function of Eq.~\ref{eq:hash} guarantee that the same ground action, in different states, will have the same cost estimates.

\beauty and \abeauty were implemented as search algorithms in \emph{PlanDEM} (Planning with Dynamically Estimated Action Models~\cite{plandem}.
a C++ planner that extends Fast Downward (FD)~\cite{helmert-jair2006} (v20.06).
All experiments were run on an Intel i7-1165G7 CPU (2.8GHz), with 32GB of RAM, in Linux.
We also implemented \emph{Estimation-time Indifferent UCS} (\eiucs), a UCS
algorithm that uses the most accurate estimate on each edge it encounters, to serve as a baseline.  
For every problem instance we ran \eiucs, \beauty with $l_{est}=l_{prune}=\infty$, and two versions of \abeauty---\abeauty-2 and \abeauty-10---with maximal number of 2 and 10 iterations, resp. We emphasize that all these algorithms are guaranteed to achieve optimal solutions. We report the results from problem instances which all algorithms solved successfully, i.e., found optimal solutions, within 5 minutes.

\subsection{\beauty vs. \eiucs}\label{subsec:beauty_exp}
%\vspace{-6pt}
%\paragraph{\beauty vs. \eiucs}
%\begin{enumerate}
%	\item Discuss the basic practical differences between the two algorithms
%	\item Compare results (search and estimation efforts) for the default case of $l_{est}=l_{prune}=\infty$
%	\item Elaborate on high variance, and rule out determining factors (plan length, plan cost, search effort). What is still to be determined in the future (graph structure, relations between action costs and estimator sets)
%	\item Emphasize that \beautyps does nothing in the default case
%\end{enumerate}
We begin by contrasting \beauty and \eiucs, to examine the effectiveness of \beauty in avoiding unnecessary expensive estimations.  
\beauty is only guaranteed optimal if its two hyper-parameters, $l_{est}$, $l_{prune}$, are greater than $L^*$, which is unknown a-priori.
Thus, to ensure a fair comparison, we set $l_{est}=l_{prune}=\infty$ for all the runs of \beauty (that are not part of the anytime framework).
Using these settings, the only difference between \beauty and \eiucs is the condition $\tilde{g}_l< g_l(s)$  in the estimation loop (line 15 in Alg.~\ref{alg:beauty}) that prevents applying further estimators when an alternative path with lower $g$-value is already known. In contrast, \eiucs ignores estimator time, always computing the tightest lower bound possible for every edge.  
Hence, the two algorithms follow the exact same search mechanism (i.e., identical node expansion order), and may only differ in the numbers and types of the estimators applied. Of specific interest is the difference in \emph{expensive} third-layer estimators usage.  
Note that under this setting \beautyps has nothing to improve, as the solution path is already fully estimated. % , so it only sets $Opt \gets true$ and $\bar{l}^* \gets l(\pi)$.

We denote by $L_3$ the number of third-layer estimators applied during search. The results are summarized below: 
\begin{itemize}
%	\item The ratio $r_{L_2}:=L_2$(\beauty)$/L_2$(\eiucs) had average of 61.9\% (stddev 10.53\%), median 61.08\%, with overall range spanning 34.75\% to 90.65\%.
	\item The ratio $r_{L_3}:=L_3$(\beauty)$/L_3$(\eiucs) had average of 60.82\% (stddev 11.57\%), median 60.88\%, with overall range spanning 24.4\% to 88.48\%.
	\item Whenever \beauty did not apply a third-layer estimator for an edge, it used on average a second-layer estimator in 0.51\% of the cases. Namely, in these cases estimation time was almost always dramatically reduced.
\end{itemize}

%We are able to rule out the following factors: solution length, the value of $l^*$, and search effort metrics such as the number of expanded nodes. % unfortunately there is no room to explain this observation

\begin{table*}[htb]
	\centering
	\begin{tabular}{r r r r r r r r} 
		Domain & \#Prob	 & $r_{L_3}$(Beauty) & $r_{\text{exp}}$(Beauty) & $r_{L_3}$(Any-2) & $r_{\text{exp}}$(Any-2) & $r_{L_3}$(Any-10) & $r_{\text{exp}}$(Any-10)\\
		\hline
		Barman      & 495 & 58.23$\pm$4.52 & 100$\pm$0 & 49.27$\pm$13.29 & 189.1$\pm$20.13  & 48.91$\pm$13.39 & 837.32$\pm$136.03 \\ 
		Caldera     & 72  & 83.25$\pm$3.01 & 100$\pm$0 & 58.48$\pm$8.08  & 176.42$\pm$9.72  & 57.88$\pm$8.42  & 905.15$\pm$90.99  \\
		Cavediving  & 54  & 70.77$\pm$0.78 & 100$\pm$0 & 59.26$\pm$3.33  & 200$\pm$0        & 59.26$\pm$3.33  & 981.48$\pm$47.88  \\
		Elevators   & 27  & 28.81$\pm$3.23 & 100$\pm$0 & 10.13$\pm$6.9   & 145.45$\pm$27.48 & 6.4$\pm$5.18    & 724.26$\pm$210.48 \\
		Floortile   & 36  & 54.83$\pm$0.76 & 100$\pm$0 & 45.13$\pm$7.51  & 183.53$\pm$13.31 & 44.6$\pm$7.68   & 890.43$\pm$100.06 \\
		Parcprinter & 36  & 83.12$\pm$2.62 & 100$\pm$0 & 25.02$\pm$11.24 & 136.34$\pm$15.58 & 22.38$\pm$9.93  & 810.26$\pm$98.03  \\ 
		Scanalyzer  & 18  & 48.18$\pm$1.65 & 100$\pm$0 & 48.16$\pm$1.66  & 200$\pm$0        & 48.16$\pm$1.66  & 994.44$\pm$23.57  \\
		Settlers    & 36  & 71.87$\pm$2.22 & 100$\pm$0 & 40.88$\pm$13.61 & 177.87$\pm$20.82 & 35.34$\pm$15.16 & 692.36$\pm$142.32 \\
		Sokoban     & 36  & 52.24$\pm$0.9  & 100$\pm$0 & 49.2$\pm$2.44   & 196.34$\pm$4.2   & 48.89$\pm$2.68  & 934.51$\pm$94.83  \\
		Tetris      & 45  & 63.3$\pm$4.74  & 100$\pm$0 & 41.91$\pm$7.27  & 180.3$\pm$10.41  & 41.09$\pm$8.37  & 907.11$\pm$126.24 \\
		Transport   & 41  & 47.25$\pm$4.09 & 100$\pm$0 & 17.53$\pm$8.76  & 144.92$\pm$20.83 & 16.01$\pm$8.73  & 760.46$\pm$132.08 \\
		Woodworking & 18  & 61.39$\pm$1.54 & 100$\pm$0 & 44.35$\pm$6.09  & 185.21$\pm$8.7   & 37.95$\pm$6.33  & 816$\pm$183.54    \\
		\textbf{All domains} & \textbf{914} & \textbf{60.82$\pm$11.57}& \textbf{100$\pm$0} & \textbf{46.03$\pm$15.75} & \textbf{182.67$\pm$23.66} & \textbf{45.13$\pm$16.37} & \textbf{849.65$\pm$142.31} \\
		
	\end{tabular}
	\caption{Summarized performance data of \beauty$(\infty,\infty)$, \abeauty-2 and \abeauty-10 (written as Beauty, Any-2 and Any-10 for brevity) relative to \eiucs, with breakdown by domains. For each algorithm and domain two entries are presented with average~$\pm$~standard deviation in percentage: the ratio of third-layer estimator usage $r_{L_3}$(Alg)$:=L_3$(Alg)$/L_3$(\eiucs) and the ratio of expanded nodes $r_{\text{exp}}$(Alg)$:=$expanded(Alg)$/$expanded(\eiucs).}
	\label{table:results}
\end{table*}

\noindent
Table~\ref{table:results} reports the results for all algorithms, compared to \eiucs. The results are grouped by domain (domains listed by row---see caption for column explanation).  The table shows (third column, total for all domains in the last row) that roughly 40\% (100-60.82) of the expensive estimations are avoided, on average. There is high variance, whose causes remain unknown for now.  

\subsection{\abeauty vs others}\label{subsec:any_exp}
%\vspace{-6pt}
%\paragraph{\abeauty}
%\begin{enumerate}
%	\item Discuss the implications of different thresholds on the maximal number of iterations allowed (search and estimation efforts). Specifically address expanded nodes, pruned nodes, l2 and l3 estimators used
%	\item Show bounds trends
%	\item Point out high variance for improvement of \abeauty with MAX-ITER=2 relative to both \beauty and \eiucs
%	\item Emphasize that MAX-ITER=2 seems best for the average case, at least without further knowledge on the estimator sets. Explain that intuitively, more iterations could be useful when the first obtained upper bound is relatively far from $l^*$ (This is supported by the data, but with relatively low numbers to gather strong statistics confidence). Another explaining factor is when there exists a relatively large gap between the best solutions on the open list, which then leads to optimality verification with significant savings
%	\item Point out low variance for improvement of \abeauty with MAX-ITER=10 relative to \abeauty with MAX-ITER=2. Given the two explaining factors presented above, this coincides with very informative first upper bounds and dense open list rival paths
%\end{enumerate}
We now turn to discuss \abeauty-2 and \abeauty-10. The relevant experiment results are summarized in Columns $5,6$ (\abeauty-2) and $7,8$ (\abeauty-10) of Table~\ref{table:results}.

First and foremost, the results reveal that \abeauty-2 and \abeauty-10 save roughly 54\% (100-46) and 55\% (100-45) of the most expensive estimations, compared to \eiucs. This represents an additional 15\% savings on top of \beauty.

Second, although both have relatively high standard deviations (about 16\%), they perform similarly in most domains (see below for the exception).  This can be attributed to the (typically) very informed upper bound $\bar{l}^*$ that is achieved after the first iteration, so there is little room for improvement. Indeed, the lower bound $\underline{l}^*$ typically comes very close to $L^*$ when \abeauty-10 converges, so when $l_{est}$ is set to $\bar{l}^*$ after the first iteration of \abeauty-2, it achieves an almost identical behavior as in the last iteration of \abeauty-10.

We examined more closely the domains where the savings of \abeauty-2 and \abeauty-10 vary noticeably (e.g., in the Elevators domain).  We observed that in many of these problems, the range of values for $c_{old}$, and thus also the range of values for the lower bound estimates (induced by $c_{old}$), is relatively high compared to other domains, i.e., the interval $[A,B] \subset [0,\infty)$ from which the values are taken is relatively large. This implies a less smooth distribution of costs (and estimates) over the graph edges, where it is common to have significant jumps in $g$-values between two subsequent nodes on a path.
The implication of such jumps is that it becomes easier to avoid estimation of non-relevant paths (with $g_l>l_{est}$). In the same cases of larger ranges of values, \abeauty-10 more frequently achieves improved estimation savings compared to \abeauty-2. We believe this may be due to the distribution of costs being less smooth, decreasing the likelihood that $\bar{l}^*$ ends up close to $L^*$ after the first iteration, and allowing more room for improvement in additional iterations.

Finally, Table~\ref{table:results} shows that the two algorithms consume on average roughly 1.8 and 8.5 times the number of expanded nodes of \eiucs, which is due to the search restart at every iteration. In domains where
the estimation savings are similar, it appears that two iterations may be sufficient, and will be much more efficient.  However, more generally---and recalling the abstracted run-time from earlier---this is a good example of how algorithms may increase the search operations, to save on weight computations.
For instance, if the times spent on estimation and search ($T_w$ and $T_v$ resp., Eq.~\ref{eq:efforts}) satisfy $T_w=10\times T_v$ for \eiucs and some problem $P$, then considering a typical factor two of savings in estimation time and twice the search time of \abeauty-2 on $P$, it follows that the latter achieves overall run-time $T_2=0.5\times T_w + 2 \times T_v=7\times T_v$ vs. $T_1=T_w +T_v=11\times T_v$, i.e., a reduction of $\approx 36\%$ in run-time.

Table~\ref{table:convergence} provides additional information that sheds light on the development of search and estimation metrics throughout the iterations. The table follows the iterations of \abeauty-10.
Row 2 indicates the number of times convergence to an optimal solution occurred at iteration $i$, allowing us to examine how many iterations were needed to solve the problems, on average. 
As can be seen, 50\% of the problems take less than 10 iterations, with rapid decrease from $i=9$ down to $i=4$, while the other 50\% terminate at $i=10$ or more (the maximum number of iterations in these experiments was 10).
Row 3 reveals the convergence of the lower bound obtained to the terminal value $L^*$. We can see that the rate of convergence is decaying. Row 4 further strengthen this observation, as the standard deviations are relatively low and also decaying. This motivates using a maximum threshold to avoid a very long convergence process, which could incur significant search effort overhead.

Lastly, Table~\ref{table:pruning} shows the average and standard deviation (Rows 2 and 3, respectively) of pruned nodes out of evaluated nodes, for each iteration of \abeauty-10, in percentage.
It can be seen that the average percentage of pruned nodes is monotonically non-decreasing with the iterations, from roughly 1\% at the second iteration to 26\% at the tenth iteration, which is due to the monotonically non-decreasing upper bound $l_{prune}$, that serves for pruning. Namely, as the upper bound gets tighter, pruning becomes more effective.

\begin{table*}[t]
	\centering
	\begin{tabular}{r r r r r r r r r r r} 
		Iteration $i$                       & $1$&$2$&$3$&$4$&$5$&$6$&$7$&$8$ &$9$ \\
		\hline
		Final $i$(\%)                       & 0  & 0 & 0 & 0 & 1 & 3 & 10& 16 & 20 \\ 
		$\mu(\underline{l}^*_i/L^*)$(\%)    & 40 & 63& 76& 85& 90& 94& 95& 97 & 97 \\ 
		$\sigma(\underline{l}^*_i/L^*)$(\%) & 7  & 9 & 9 & 8 & 7 & 6 & 5 & 4  & 4  \\ 
		
	\end{tabular}
	\caption{Convergence analysis of \abeauty-10, at iteration number $i$. Row 2 indicates the number of times convergence occurred at iteration $i$, Rows 3 and 4 indicate the mean $\mu$ and standard deviation $\sigma$, respectively, for the ratio of the lower bound obtained after iteration $i$ to $L^*$, where the values in Rows 2--4 are in percentages. Results are rounded to integers for ease of presentation.}
	\label{table:convergence}
\end{table*}

\begin{table*}[t]
	\centering
	\begin{tabular}{r r r r r r r r r r r} 
		Iteration $i$                     & $1$&$2$&$3$&$4$  & $5$ & $6$ & $7$ & $8$  & $9$& $10$ \\
		\hline
		$\mu(\text{pr}/\text{ev})$(\%)    & 0  & 1 & 2 & 4   & 10  & 11  & 12  & 15   & 17 & 26  \\ 
		$\sigma(\text{pr}/\text{ev})$(\%) & 0  & 5 & 7 & 10  & 16  & 16  & 17  & 18   & 19 & 22  \\ 
%$\sigma(\text{prun}/\text{eval})$		
	\end{tabular}
	\caption{Pruning analysis of \abeauty-10, at iteration number $i$. Rows 2 and 3 indicate the mean $\mu$ and standard deviation $\sigma$, respectively, for the ratio of pruned nodes out of evaluated nodes, in percentages. Results are rounded to integers for ease of presentation.}
	\label{table:pruning}
\end{table*}

\subsection{\beautyps}\label{subsec:beautyps}
%\vspace{-6pt}
%\begin{enumerate}
%	\item Demonstrate the effectiveness of the procedure in obtaining a very informative upper bound, while consuming very little estimation resources
%	\item Provide empirical evidence for its ability to verify optimality in cases where $l_{est} \le l^*$
%\end{enumerate}
Given that often, two iterations of \abeauty offered the same savings as ten iterations, yet significantly more than a single iteration, it is interesting to examine the role of \beautyps (Procedure~\ref{alg:beautyps}) in improving the results from the first iteration of \abeauty.  Recall that \beautyps  obtains the tightest possible lower bound $\bar{l}^*$ for $c(\pi)$, which can then either be interpreted as $L^*$ if $opt=true$ is returned, or as an upper bound for $L^*$ otherwise. When \beauty is called with its hyper-parameters set to $\infty$, it is optimal; \beautyps has nothing to improve. However, when it is called as part of \abeauty, the hyper-parameters are different, which gives \beautyps the potential to improve the results before the next iteration.

The results provide insight to the efficacy of this procedure. When calling \beautyps after \beauty is run with $l_{est}=0$ and $l_{prune}=\infty$ (the least informative hyper-parameters), \beautyps returns on average $\bar{l}^*=1.0082 \times L^*$, i.e., only 0.82\% higher than $L^*$, with standard deviation of 3.31\%, where in the worst case $\bar{l}^*$ was 33.33\% higher than $L^*$.
This means that just one iteration of \beauty that uses the cheapest lower bounds during the search, followed by \beautyps, typically returns a very good approximation of $L^*$ in the form of a very informed upper bound for it.
Furthermore, \beautyps utilizes only a tiny fraction of the expensive estimators, as it only estimates edges on the solution path. 
Thus, on average, \beauty with $l_{est}=0, l_{prune}=\infty$ was able to generate a very accurate approximation of the optimal solution, though at the loss of guaranteed optimality, at minimal estimation effort overhead.  

%\begin{observation}\label{cor:approximation_of_l*}
%	If one is willing to compromise optimality, a greedy approach---that finds a solution solely based on the cheapest estimates---can often yield a good approximation.
%\end{observation}
%Conservatively, we recognize that this observation could be wrong for cases that were not considered in the scope of our experiments. 
%Since we did not recognize factors that affect the fitness of this approximation, it is left for future work.
%TODO: it seems like a high priority research direction to understand when does the greedy approach probably yields a good approximation.

\subsection{Different Accuracy Levels}\label{subsec:accuracy}
Table~\ref{table:f_config} determines the accuracy range of estimators in our experiments.
Indeed, for an edge $e$, the accuracy of its first cost estimate $l_e^1$ relative to its best estimate $l_e^3$ is $l_e^1/l_e^3 = f_1/f_3$.
A high ratio of $f_1/f_3$ implies that a cheap estimator yields a good approximation of the best estimate.
It is thus interesting to test the sensitivity of the algorithms discussed in this section w.r.t. different accuracy levels.
% which are implied by different configurations of $f_1,f_2$ and $f_3$.

To that end, we ran another experiment with the same setting as described before, in four domains (Barman, Settlers, Sokoban, Tetris), and with $f_1 \in \{10, 11, 12\}, f_2 \in \{f_1+1, f_1+2, f_1+3\}, f_3 \in \{f_2+1\}$, which resulted in significantly higher ratios of $f_1/f_3$.
Specifically, the range of $f_1/f_3$ changed from $20\%-60\%$ to roughly $71.43\%-83.33\%$. 

The results of expensive estimator usage (i.e., $r_{L_3}$ of \beauty, \abeauty-2 and \abeauty-10) are almost identical to the results reported in Table~\ref{table:results}, with at most $1\%$ difference in any entry. However, the convergence of \abeauty-10 was faster, where the average number of iterations until convergence changed from 9 in the first experiment to 5.65 in the second experiment.
This suggests that relatively accurate cheap estimators do not affect the number of expensive estimations required to achieve optimality, but reduce the number of iterations necessary for convergence.
\section{Conclusions}\label{sec:conc}
%% re-iterate what we've done: novel framework, derived problems, algorithms
This paper presents a generalized framework for \textit{estimated weighted directed graphs}, where the cost of each edge can be estimated by multiple estimators, where every estimator has its own run-time and returns lower and upper bounds on the edge weight.
This allows to address novel settings of combinatorial search problems that support an explicit trade-off of search and estimation time.
We focus on the \textit{shortest path tightest lower bound} (SLB) problem, which we formally define. SLB problems involve finding a path with the tightest lower bound on the optimal cost. We present two algorithms for solving SLB problems in a guaranteed manner.  Experiments reveal
the dramatic computational savings they offer.
%% suggest future work on the defined problems, on other possible problems, on undirected graphs, and on informed search
%TODO: suggest future work to develop a meta-reasoning framework based on \abeauty that takes into account various data channels (search time, estimation time, estimator density, lower bound and upper bound) to determine how to set $l_{est}, l_{prune}$ in each iteration.
%A natural line of future work is to develop algorithms, similar in spirit to the ones presented here, for solving other variants of the shortest path problem, analogous to the problem we address here. 
%the shortest path w.r.t. upper bounds and the tightest admissible shortest path.

There are many directions for future research. We believe the performance of the algorithms can be further improved (e.g., by utilizing priors on estimation times to choose estimators across edges).  % Other algorithmic approaches can be tested as well.
We plan to investigate shortest path variants that minimize path upper-bound and $\mathcal{B}$-admissibility.
Extensions for undirected graphs and for informed search are also of significant interest. 
%, as are additional novel graph search problems that are based on estimated, rather than exact, costs. 

%, each aiming for a different trade-off between optimism and pessimism in face of uncertainty. 
%Thus a natural direction is to develop algorithms that solve new problems.

\section*{Acknowledgements}
The research was partially funded by ISF Grant \#2306/18 and BSF-NSF grant 2017764. Thanks to K. Ushi. 
This research was also supported by ISF grant \#909/23 to Ariel Felner.
Eyal Weiss is supported by the Adams Fellowship Program of the Israel Academy of Sciences and Humanities and by Bar-Ilan University's President Scholarship.

%\clearpage
% References and End of Paper
% These lines must be placed at the end of your paper
% \bibliographystyle{named}
\bibliography{short,weiss_references,graph_references}
\end{document}